\documentclass{llncs}

\usepackage{xspace}
\usepackage{color}
\usepackage{epsfig}
\usepackage{times}
\usepackage{gensymb}
\usepackage[english]{babel}
\usepackage{graphicx}
\usepackage{amssymb}
\usepackage{multirow}
\usepackage{amsmath}

\title{Increasing-Chord Graphs On Point Sets\thanks{Work partially supported by the Australian Research Council (grant DE140100708).}}

\date{}

\author{Hooman Reisi Dehkordi$^1$, Fabrizio Frati$^2$, Joachim Gudmundsson$^2$
\institute{
$^1$ School of Information Technologies -- Monash University\\
\email{hooman.dehkordi@monash.edu}\\
$^2$ School of Information Technologies -- The University of Sydney\\
\email{\{fabrizio.frati,joachim.gudmundsson\}@sydney.edu.au}}}

\newcommand{\remove}[1]{}

\renewenvironment{proof}
{{\bf Proof:}}{\hspace*{\fill}$\Box$\par\vspace{2mm}}

\begin{document}

\maketitle

\begin{abstract}
We tackle the problem of constructing increasing-chord graphs spanning point sets. We prove that, for every point set $P$ with $n$ points, there exists an increasing-chord planar graph with $O(n)$ Steiner points spanning $P$. Further, we prove that, for every convex point set $P$ with $n$ points, there exists an increasing-chord graph with $O(n \log n)$ edges (and with no Steiner points) spanning $P$.
\end{abstract}

\section{Introduction} \label{se:introduction}

A {\em proximity graph} is a geometric graph that can be constructed from a point set by connecting points that are ``close'', for some local or global definition of proximity. Proximity graphs constitute a topic of research in which the areas of graph drawing and computational geometry nicely intersect. A typical graph drawing question in this topic asks to characterize the graphs that can be represented as a certain type of proximity graphs. A typical computational geometry question asks to design an algorithm to construct a proximity graph spanning a given point set.

Euclidean minimum spanning trees and Delaunay triangulations are famous examples of proximity graphs. Given a point set $P$, a {\em Euclidean minimum spanning tree} (MST) of $P$ is a geometric tree with $P$ as vertex set and with minimum total edge length; the {\em Delaunay triangulation} of $P$ is a triangulation $T$ such that no point in $P$ lies inside the circumcircle of any triangle of $T$. From a computational geometry perspective, given a point set $P$ with $n$ points, an MST of $P$ with maximum degree five exists~\cite{MonmaS92} and can be constructed in $O(n \log n)$ time~\cite{bcko-cgta-08}; also, the Delaunay triangulation of $P$ exists and can be constructed in $O(n \log n)$ time~\cite{bcko-cgta-08}. From a graph drawing perspective, every tree with maximum degree five admits a representation as an MST~\cite{MonmaS92} and it is NP-hard to decide whether a tree with maximum degree six admits such a representation~\cite{EadesW96}; also, characterizing the class of graphs that can be represented as Delaunay triangulations is a deeply studied question, which still eludes a clear answer; see, e.g.,~\cite{dv-aptg-96,ds-gtcidr-96}. Refer to the excellent survey by Liotta~\cite{gd-handbook} for more on proximity graphs.

While proximity graphs have constituted a frequent topic of research in graph drawing and computational geometry, they gained a sudden peak in popularity even outside these communities in 2004, when Papadimitriou {\em et al.}~\cite{conf/mobicom/RaoPSS03} devised an elegant routing protocol that works effectively in all the networks that can be represented as a certain type of proximity graphs, called {\em greedy graphs}. For two points $p$ and $q$ in the plane, denote by $\overline{pq}$ the straight-line segment having $p$ and $q$ as end-points, and by $|\overline{pq}|$ the length of $\overline{pq}$. A geometric path $(v_1,\dots,v_n)$ is {\em greedy} if $|\overline{v_{i+1}v_n}|<|\overline{v_{i}v_n}|$, for every $1\leq i\leq n-1$. A geometric graph $G$ is {\em greedy} if, for every ordered pair of vertices $u$ and $v$, there exists a greedy path from $u$ to $v$ in $G$. A result related to our paper is that, for every point set $P$, the Delaunay triangulation of $P$ is a greedy graph~\cite{PapadimitriouR05}.

In this paper we study {\em self-approaching} and {\em increasing-chord graphs}, that are types of proximity graphs defined by Alamdari {\em et al.}~\cite{acglp-sag-12}. A geometric path ${\cal P}=(v_1,\dots,v_n)$ is {\em self-approaching} if, for every three points $a$, $b$, and $c$ in this order on ${\cal P}$ from $v_1$ to $v_n$ (possibly $a$, $b$, and $c$ are internal to segments of ${\cal P}$), it holds that $|\overline{bc}|<|\overline{ac}|$. A geometric graph $G$ is {\em self-approaching} if, for every ordered pair of vertices $u$ and $v$, $G$ contains a self-approaching path from $u$ to $v$; also, $G$ is {\em increasing-chord} if, for every pair of vertices $u$ and $v$, $G$ contains a path between $u$ and $v$ that is self-approaching both from $u$ to $v$ and from $v$ to $u$; thus, an increasing-chord graph is also self-approaching. The study of self-approaching and increasing-chord graphs is motivated by their relationship with greedy graphs (a self-approaching graph is also greedy), and by the fact that such graphs have a small geometric dilation, namely at most 5.3332~\cite{ikl-sac-99} (self-approaching graphs) and at most 2.094~\cite{r-cic-94} (increasing-chord graphs).

Alamdari {\em et al.} showed: (i) how to test in linear time whether a path in $\mathbb R^2$ is self-approaching; (ii) a characterization of the class of self-approaching trees; and (iii) how to construct, for every point set $P$ with $n$ points in $\mathbb R^2$, an increasing-chord graph that spans $P$ and uses $O(n)$ Steiner points.

In this paper we focus our attention on the problem of constructing increasing-chord graphs spanning given point sets in $\mathbb R^2$. We prove two main results.

\begin{itemize}
\item We show that, for every point set $P$ with $n$ points, there exists an increasing-chord planar graph with $O(n)$ Steiner points spanning $P$. This answers a question of Alamdari {\em et al.}~\cite{acglp-sag-12} and improves upon their result (iii) above, since our increasing-chord graphs are planar and contain increasing-chord paths between every pair of points, including the Steiner points (which is not the case for the graphs in~\cite{acglp-sag-12}). It is interesting that our result is achieved by studying Gabriel triangulations, which are proximity graphs strongly related to Delaunay triangulations (a Gabriel triangulation of a point set $P$ is a subgraph of the Delaunay triangulation of $P$). It has been proved in~\cite{acglp-sag-12} that Delaunay triangulations are not, in general, self-approaching.
\item We show that, for every convex point set $P$ with $n$ points, there exists an increasing-chord graph that spans $P$ and that has $O(n \log n)$ edges (and no Steiner points).
\end{itemize}

\section{Definitions and Preliminaries} \label{se:preliminaries}

A {\em geometric graph} $(P,S)$ consists of a point set $P$ in the plane and of a set $S$ of straight-line segments (called {\em edges}) between points in $P$. A geometric graph is {\em planar} if no two of its edges cross. A planar geometric graph partitions the plane into connected regions called {\em faces}. The bounded faces are {\em internal} and the unbounded face is the {\em outer face}. A geometric planar graph is a {\em triangulation} if every internal face is delimited by a triangle and the outer face is delimited by a convex polygon.

Let $p$, $q$, and $r$ be points in the plane. We denote by $\angle{pqr}$ the angle defined by a clockwise rotation around $q$ bringing $\overline{pq}$ to coincide with $\overline{qr}$.

A {\em convex combination} of a set of points $P=\{p_1,\dots,p_k\}$ is a point $\sum \alpha_i p_i$ where $\sum \alpha_i=1$ and $\alpha_i\geq 0$ for each $1\leq i\leq k$. The {\em convex hull} ${\cal H}_P$ of $P$ is the set of points that can be expressed as a convex combination of the points in $P$. A {\em convex point set} $P$ is such that no point is a convex combination of the others. Let $P$ be a convex point set and $\vec d$ be a directed straight line not orthogonal to any line through two points of $P$. Order the points in $P$ as their projections appear on $\vec d$; then, the {\em minimum point} and the {\em maximum point} of $P$ with respect to $\vec d$ are the first and the last point in such an ordering. We say that $P$ is  {\em one-sided with respect to $\vec d$} if the minimum and the maximum point of $P$ with respect to $\vec d$ are consecutive along the border of ${\cal H}_P$. See Fig.~\ref{fig:one-sided-x}. A {\em one-sided convex point set} is a convex point set that is one-sided with respect to some directed straight line $\vec d$. The proof of our first lemma shows an algorithm to construct an increasing-chord planar graph spanning a one-sided convex point set.

\begin{figure}[tb]
\begin{center}
\mbox{\includegraphics[scale=0.32]{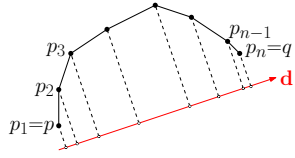}}
\caption{A convex point set that is one-sided with respect to a directed straight line $\vec d$.}
\label{fig:one-sided-x}
\end{center}
\end{figure}

\begin{lemma} \label{lemma:half_convex_planar_graph}
Let $P$ be any one-sided convex point set with $n$ points. There exists an increasing-chord planar graph spanning $P$ with $2n-3$ edges.
\end{lemma}

\begin{proof}
Assume that $P$ is one-sided with respect to the positive $x$-axis ${\vec x}$. Such a condition can be met after a suitable rotation of the Cartesian axes. Let $\{p_1,p_2,\dots,p_n\}$ be the points in $P$, ordered as their projections appear on ${\vec x}$.

We show by induction on $n$ that an increasing-chord planar graph $G$ spanning $P$ exists, in which all the edges on the border of ${\cal H}_P$ are in $G$. If $n=2$ then the graph with a single edge $\overline{p_1p_2}$ is an increasing-chord planar graph spanning $P$. Next, assume that $n>2$ and let $p_j$ be a point with largest $y$-coordinate in $P$ (possibly $j=1$ or $j=n$). Point set $Q=P\setminus\{p_j\}$ is convex, one-sided with respect to ${\vec x}$, and has $n-1$ points. By induction, there exists an increasing-chord planar graph $G'$ spanning $Q$ in which all the edges on the border of ${\cal H}_{Q}$ are in $G'$. Let $G$ be the graph obtained by adding vertex $p_j$ and edges $\overline{p_{j-1}p_j}$ and $\overline{p_jp_{j+1}}$ to $G'$. We have that $G$ is planar, given that $G'$ is planar and that edges $\overline{p_{j-1}p_j}$ and $\overline{p_jp_{j+1}}$ are on the border of ${\cal H}_{P}$. Further, all the edges on the border of ${\cal H}_{P}$ are in $G$. Moreover, $G$ contains an increasing-chord path between every pair of points in $Q$, by induction; also, $G$ contains an increasing-chord path between $p_j$ and every point $p_i$ in $Q$, as one of the two paths on the border of ${\cal H}_{P}$ connecting $p_j$ and $p_i$ is both $x$- and $y$-monotone, and hence increasing-chord by the results in~\cite{acglp-sag-12}. Finally, $G$ is a maximal outerplanar graph, hence it has $2n-3$ edges.
\end{proof}

The {\em Gabriel graph} of a point set $P$ is the geometric graph that has an edge $\overline{pq}$ between two points $p$ and $q$ if and only if the closed disk whose diameter is $\overline{pq}$ contains no point of $P\setminus\{p,q\}$ in its interior or on its boundary. A {\em Gabriel triangulation} is a triangulation that is the Gabriel graph of its point set $P$. We say that a point set $P$ {\em admits} a Gabriel triangulation if the Gabriel graph of $P$ is a triangulation. A triangulation is a Gabriel triangulation if and only if every angle of a triangle delimiting an internal face is acute~\cite{gs-nsagva-69}. See~\cite{gs-nsagva-69,gd-handbook,ms-pgg-80} for more properties about Gabriel graphs.

In Section~\ref{se:steiner} we will prove that every Gabriel triangulation is increasing-chord. A weaker version of the converse is also true, as proved in the following.

\begin{lemma} \label{lemma:gabriel_edge_increasing_chord}
Let $P$ be a set of points and let $G(P,S)$ be an increasing-chord graph spanning $P$. Then all the edges of the Gabriel graph of $P$ are in $S$.
\end{lemma}

\begin{proof}
Suppose, for a contradiction, that there exists an increasing-chord graph $G(P,S)$ and an edge $\overline{uv}$ of the Gabriel graph of $P$ such that $\overline{uv}\notin S$. Then, consider any increasing-chord path ${\cal P}=(u=w_1,w_2,\dots,w_k=v)$ in $G$. Since $\overline{uv}\notin S$, it follows that $k>2$. Assume w.l.o.g. that $w_1$, $w_2$, and $w_k$ appear in this clockwise order on the boundary of triangle $(w_1,w_2,w_k)$. Since the closed disk with diameter $\overline{uv}$ does not contain any point in its interior or on its boundary, it follows that $\angle{w_k w_2 w_1}<90^\circ$. If $\angle{w_2 w_1 w_k}\geq 90^\circ$, then $|w_1w_k|<|w_2w_k|$, a contradiction to the assumption that ${\cal P}$ is increasing-chord. If $\angle{w_2w_1w_k}<90^\circ$, then the altitude of triangle $(w_1,w_2,w_k)$ incident to $w_k$ hits $\overline{w_1 w_2}$ in a point $h$. Hence, $|hw_k|<|w_2w_k|$, a contradiction to the assumption that ${\cal P}$ is increasing-chord which proves the lemma.
\end{proof}

\section{Planar Increasing-Chord Graphs with Few Steiner Points} \label{se:steiner}

We show that, for any point set $P$, one can construct an increasing-chord planar graph $G(P',S)$ such that $P\subseteq P'$ and $|P'|\in O(|P|)$. Our result has two ingredients. The first one is that Gabriel triangulations are increasing-chord graphs. The second one is a result of Bern {\em et al.}~\cite{beg-pgmg-94} stating that, for any point set $P$, there exists a point set $P'$ such that $P\subseteq P'$, $|P'|\in O(|P|)$, and $P'$ admits a Gabriel triangulation. Combining these two facts proves our main result. The proof that Gabriel triangulations are increasing-chord graphs consists of two parts. In the first one, we prove that geometric graphs having a {\em $\theta$-path} between every pair of points are increasing-chord. In the second one, we prove that in every Gabriel triangulation there exists a $\theta$-path between every pair of points.

%Consider a point set $V$ in the plane and assume that there exists a Gabriel triangulation $T(V,E)$ of $S$.

We introduce some definitions. The {\em slope} of a straight-line segment $\overline{uv}$ is the angle spanned by a clockwise rotation around $u$ that brings $\overline{uv}$ to coincide with the positive $x$-axis. Thus, if $\theta$ is the slope of $\overline{uv}$, then $\theta+k\cdot 360^{\circ}$ is also the slope of $\overline{uv}$, $\forall k\in \mathbb{Z}$. A straight-line segment $\overline{uv}$ is a {\em $\theta$-edge} if its slope is in the interval $[\theta-45^\circ;\theta+45^\circ]$. Also, a geometric path ${\cal P}=(p_1, \dots, p_k)$ is a {\em $\theta$-path} from $p_1$ to $p_k$ if $\overline{p_ip_{i+1}}$ is a $\theta$-edge, for every $1\leq i\leq k-1$. Consider a point $a$ on a $\theta$-path ${\cal P}$ from $p_1$ to $p_k$. Then, the subpath ${\cal P}_a$ of ${\cal P}$ from $a$ to $p_k$ is also a $\theta$-path. Moreover, denote by $W_\theta(a)$ the closed wedge with an angle of $90^\circ$ incident to $a$ and whose delimiting lines have slope $\theta-45^\circ$ and $\theta+45^\circ$; then ${\cal P}_a$ is contained in $W_\theta(a)$ (see Fig.~\ref{fig:wedges}). We have the following:

\begin{figure}[tb]
\begin{center}
\mbox{\includegraphics[scale=0.35]{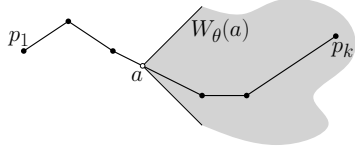}}
\caption{Wedge $W_\theta(a)$ contains path ${\cal P}_a$.}
\label{fig:wedges}
\end{center}
\end{figure}

\begin{lemma} \label{le:theta-path-approaching}
Let $\cal P$ be a $\theta$-path from $p_1$ to $p_k$. Then, $\cal P$ is increasing-chord.
\end{lemma}

\begin{proof}
Lemma 3 in~\cite{ikl-sac-99} states the following (see also~\cite{aaiklr-gsac-01}): A curve $\cal C$ with end-points $p$ and $q$ is self-approaching from $p$ to $q$ if and only if, for every point $a$ on $\cal C$, there exists a closed wedge with an angle of $90^\circ$ incident to $a$ and containing the part of $\cal C$ between $a$ and $q$.
By definition of $\theta$-path, for every point $a$ on $\cal P$, the closed wedge $W_\theta(a)$ with an angle of $90^\circ$ incident to $a$ and whose delimiting lines have slope $\theta-45^\circ$ and $\theta+45^\circ$ contains the subpath ${\cal P}_a$ of $\cal P$ from $a$ to $p_k$. Hence, by Lemma 3 in~\cite{ikl-sac-99}, $\cal P$ is self-approaching from $p_1$ to $p_k$.
An analogous proof shows that $\cal P$ is self-approaching from $p_k$ to $p_1$, given that $\cal P$ is a $(\theta+180^{\circ})$-path from $p_k$ to $p_1$.
\remove{We prove that $\cal P$ is self-approaching from $p_1$ to $p_k$. The proof that $\cal P$ is self-approaching from $p_k$ to $p_1$ is analogous, as $\cal P$ is a $(\theta+180^{\circ})$-path from $p_k$ to $p_1$.

Consider any three points $a$, $b$, and $c$ on $\cal P$ in this order from $p_1$ to $p_k$. Assume that $a$, $b$, and $c$ appear in this counter-clockwise order along triangle $(a,b,c)$, the case in which they appear in clockwise order $a$, $b$, and $c$ is analogous. In order to prove that $|\overline{ac}|>|\overline{bc}|$ it suffices to prove that $\angle{abc}\geq 90^\circ$. Since $c$ is in $W_{\theta}(b)$ and since $b$ is in $W_{\theta}(a)$, we have that $\angle{abc}\geq \angle{abc'}$, where $c'$ is the intersection point of $\overline{ac}$ with the line delimiting $W_{\theta}(b)$ with slope $\theta+45^\circ$. Moreover, $\angle{abc'}\geq 90^\circ$, given that $\angle{abc'}$ is the supplementary  of the angle obtained by counter-clockwise rotating $\overline{ab}$ until it coincides with the line delimiting $W_{\theta}(a)$ with slope $\theta+45^\circ$. The latter angle is smaller than or equal to $90^\circ$, given that it is not larger than the angle of wedge $W_{\theta}(a)$. This proves that $\cal P$ is self-approaching from $p_1$ to $p_k$, and hence it proves the lemma.
}
\end{proof}

\remove{
%From the previous lemma, we directly obtain the following.
%\begin{corollary} \label{cor:theta-is-increasing}
%Let $G$ be a geometric graph on a point set $P$ and assume that, for every ordered pair of points %$(s,t)$, with $s,t \in P$, there exists an angle $\theta_{s,t}$ such that $G$ contains a% %$\theta_{s,t}$-path from $s$ to $t$. Then, $G$ is increasing-chord.
%\end{corollary}
}
We now prove that Gabriel triangulations contain $\theta$-paths.

\begin{lemma} \label{le:gabriel-is-theta}
Let $G$ be a Gabriel triangulation on a point set $P$. For every two points $s,t\in P$, there exists an angle $\theta$ such that $G$ contains a $\theta$-path from $s$ to $t$.
\end{lemma}

\begin{proof}
Consider any two points $s,t\in P$. Clockwise rotate $G$ of an angle $\phi$ so that $y(s)=y(t)$ and $x(s)<x(t)$. Observe that, if there exists a $\theta$-path from $s$ to $t$ after the rotation, then there exists a $(\theta+\phi)$-path from $s$ to $t$ before the rotation.
%We assume that $G$ does not contain edge $\overline{st}$, as otherwise the slope of such an edge is desired angle $\theta_{s,t}$.

A $\theta$-path $(p_1,\dots,p_k)$ in $G$ is {\em maximal} if there is no $z\in P$ such that $\overline{p_kz}$ is a $\theta$-edge. For every maximal $\theta$-path ${\cal P}=(p_1, \ldots , p_k)$ in $G$, $p_k$ lies on the border of ${\cal H}_P$. Namely, assume the converse, for a contradiction. Since $G$ is a Gabriel triangulation, the angle between any two consecutive edges incident to an internal vertex of $G$ is smaller than $90^\circ$, thus there is a $\theta$-edge incident to $p_k$.  This contradicts the maximality of $\cal P$.  A maximal $\theta$-path $(s=p_1, \ldots , p_k)$ is {\em high} if either (a) $y(p_k)>y(t)$ and $x(p_k)<x(t)$, or (b) $\overline{p_i p_{i+1}}$ intersects the vertical line through $t$ at a point above $t$, for some $1\leq i\leq k-1$. Symmetrically, a maximal $\theta$-path $(s=p_1, \ldots , p_k)$ is {\em low} if either (a) $y(p_k)<y(t)$ and $x(p_k)<x(t)$, or (b) $\overline{p_i p_{i+1}}$ intersects the vertical line through $t$ at a point below $t$, for some $1\leq i\leq k-1$. High and low $(\theta+180^\circ)$-paths starting at $t$ can be defined analogously. The proof of the lemma consists of two main claims.

{\bf Claim 1.} If a maximal $\theta$-path ${\cal P}_s$ starting at $s$ and a maximal $(\theta+180^\circ)$-path ${\cal P}_t$ starting at $t$ exist such that ${\cal P}_s$ and ${\cal P}_t$ are both high or both low, for some $-45^\circ\leq \theta \leq 45^{\circ}$, then there exists a $\theta$-path in $G$ from $s$ to $t$.

{\bf Claim 2.} For some $-45^\circ\leq \theta \leq 45^{\circ}$, there exist a maximal $\theta$-path ${\cal P}_s$ starting at $s$ and a maximal $(\theta+180^\circ)$-path ${\cal P}_t$ starting at $t$ that are both high or both low.

Observe that Claims 1 and 2 imply the lemma.

We now prove Claim~1. Suppose that $G$ contains a maximal high $\theta$-path ${\cal P}_s$ starting at $s$ and a maximal high $(\theta+180^\circ)$-path ${\cal P}_t$ starting at $t$, for some $-45^\circ\leq \theta \leq 45^{\circ}$. If ${\cal P}_s$ and ${\cal P}_t$ share a vertex $v\in P$, then the subpath of ${\cal P}_s$ from $s$ to $v$ and the subpath of ${\cal P}_t$ from $v$ to $t$ form a $\theta$-path in $G$ from $s$ to $t$. Thus, it suffices to show that ${\cal P}_s$ and ${\cal P}_t$ share a vertex. For a contradiction assume the converse. Let $p_s$ and $p_t$ be the end-vertices of ${\cal P}_s$ and ${\cal P}_t$ different from $s$ and $t$, respectively. Recall that $p_s$ and $p_t$ lie on the border of ${\cal H}_P$. Denote by ${\vec l}_s$ and ${\vec l}_t$ the vertical half-lines starting at $s$ and $t$, respectively, and directed towards increasing $y$-coordinates; also, denote by $q_s$ and $q_t$ the intersection points of ${\vec l}_s$ and ${\vec l}_t$ with the border of ${\cal H}_P$, respectively. Finally, denote by $Q$ the curve obtained by clockwise following the border of ${\cal H}_P$ from $q_s$ to $q_t$.

Assume that $x(p_s)\geq x(t)$, as in Fig.~\ref{fig:gabriel}(a). Path ${\cal P}_s$ starts at $s$ and passes through a point $r_s$ on ${\vec l}_t$ (possibly $r_s=q_t$), given that  $x(p_s)\geq x(t)$. Path ${\cal P}_t$ starts at $t$ and either passes through a point $r_t$ on ${\vec l}_s$, or ends at a point $p_t$ on $Q$, depending on whether $x(p_t)\leq x(s)$ or $x(p_t)> x(s)$, respectively. Since ${\cal P}_s$ is $x$-monotone and lies in ${\cal H}_P$, it follows that $r_t$ and $p_t$ are above or on ${\cal P}_s$; also, $t$ is below ${\cal P}_s$ given that ${\cal P}_s$ is a high path. It follows ${\cal P}_s$ and ${\cal P}_t$ intersect, hence they share a vertex given that $G$ is planar.

\begin{figure}[tb]
\begin{center}
\begin{tabular}{c c c}
\mbox{\includegraphics[scale=0.45]{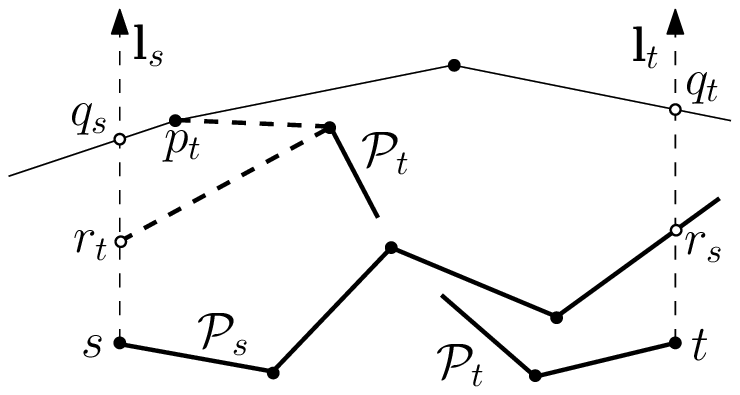}} & \hspace{5mm}
\mbox{\includegraphics[scale=0.45]{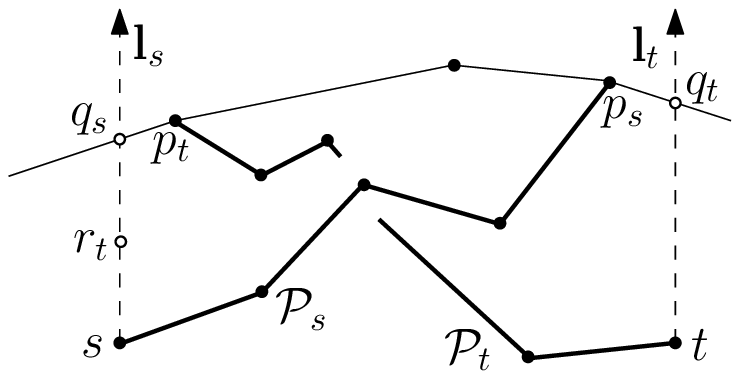}} & \hspace{5mm}
\mbox{\includegraphics[scale=0.45]{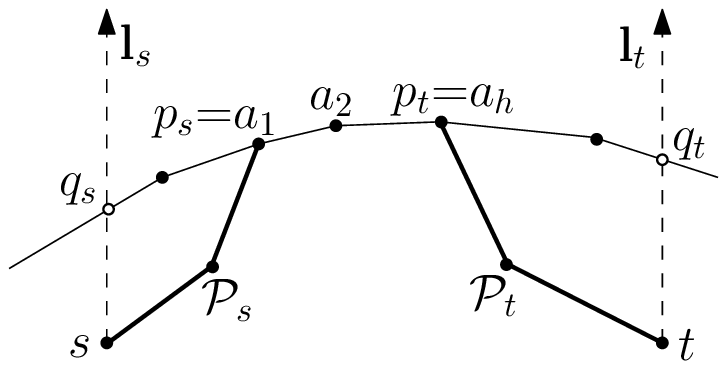}}\\
(a) \hspace{5mm} & \hspace{5mm} (b) & \hspace{5mm} (c)
\end{tabular}
\caption{Paths ${\cal P}_s$ and ${\cal P}_t$ intersect if: (a) $x(p_s)\geq x(t)$, (b) $x(s)<x(p_t)<x(p_s)<x(t)$, and (c) $x(s)<x(p_s)<x(p_t)<x(t)$.}
\label{fig:gabriel}
\end{center}
\end{figure}

Analogously, if $x(p_t)\leq x(s)$, then ${\cal P}_s$ and ${\cal P}_t$ share a vertex.

If $x(p_t)= x(p_s)$, then ${\cal P}_s\cup {\cal P}_t$ is a $\theta$-path from $s$ to $t$.

Next, if $x(s)<x(p_t)<x(p_s)<x(t)$, as in Fig.~\ref{fig:gabriel}(b), then the end-points of ${\cal P}_s$ and ${\cal P}_t$ alternate along the boundary of the region $R$ that is the intersection of ${\cal H}_P$, of the half-plane to the right of ${\vec l}_s$, and of the half-plane to the left of ${\vec l}_t$. Since ${\cal P}_s$ and ${\cal P}_t$ are $x$-monotone, they lie in $R$, thus they intersect, and hence they share a vertex.

Finally, assume that $x(s)<x(p_s)< x(p_t)<x(t)$, as in Fig.~\ref{fig:gabriel}(c). Let $a_1, \ldots , a_h$ be the clockwise order of the points along $Q$, starting at $p_s=a_1$ and ending at $a_h=p_t$. By the assumption $x(p_s)< x(p_t)$ we have $h\geq 2$. We prove that $\overline{a_1a_2}$ is a $\theta$-edge. Suppose, for a contradiction, that $\overline{a_1a_2}$ is not a $\theta$-edge. Since the slope of $\overline{a_1a_2}$ is larger than $-90^\circ$ and smaller than $90^\circ$, it is either larger than $\theta+45^\circ$ and smaller than $90^{\circ}$, or it is larger than $-90^\circ$ and smaller than $\theta-45^{\circ}$. First, assume that the slope of $\overline{a_1a_2}$ is larger than $\theta+45^\circ$ and smaller than $90^{\circ}$, as in Fig.~\ref{fig:gabriel-2}(a). Since the slope of $\overline{sa_1}$ is between $\theta-45^\circ$ and $\theta+45^\circ$, it follows that $a_1$ is below the line composed of $\overline{sa_2}$ and $\overline{a_2t}$, which contradicts the assumption that $a_1$ is on $Q$. Second, if the slope of $\overline{a_1a_2}$ is larger than $-90^\circ$ and smaller than $\theta-45^{\circ}$, then we distinguish two further cases. In the first case, represented in Fig.~\ref{fig:gabriel-2}(b), the slope of $\overline{a_1t}$ is larger than $\theta-45^{\circ}$, hence $a_2$ is below the line composed of $\overline{sa_1}$ and $\overline{a_1t}$, which contradicts the assumption that $a_2$ is on $Q$. In the second case, represented in Fig.~\ref{fig:gabriel-2}(c), the slope of $\overline{a_1t}$ is in the interval $[-90^\circ;\theta-45^\circ]$. It follows that the slope of $\overline{ta_1}$ is in the interval $[90^\circ;\theta+135^\circ]$; since the slope of $\overline{ta_h}$ is smaller than the one of $\overline{ta_1}$, we have that ${\cal P}_t$ is not a $(\theta+180^\circ)$-path. This contradiction proves that $\overline{a_1a_2}$ is a $\theta$-edge. However, this contradicts the assumption that ${\cal P}_s$ is a maximal $\theta$-path, and hence concludes the proof of Claim~1.

\begin{figure}[tb]
\begin{center}
\begin{tabular}{c c c}
\mbox{\includegraphics[scale=0.45]{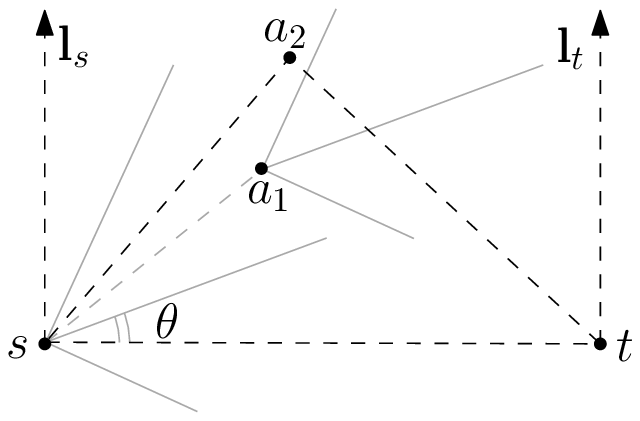}} & \hspace{5mm}
\mbox{\includegraphics[scale=0.45]{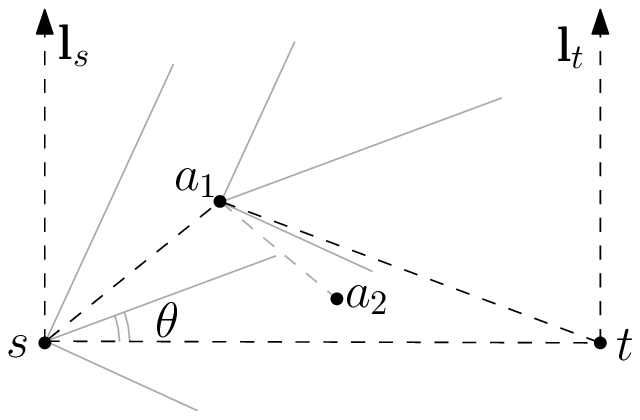}} & \hspace{5mm}
\mbox{\includegraphics[scale=0.45]{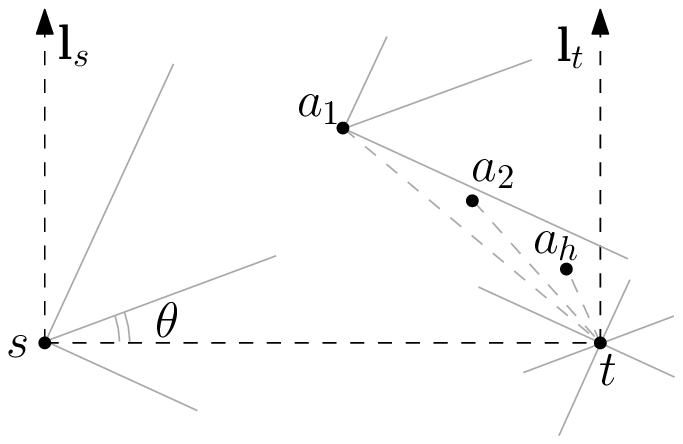}}\\
(a) \hspace{5mm} & \hspace{5mm} (b) & \hspace{5mm} (c)
\end{tabular}
\caption{Illustration for the proof that $\overline{a_1a_2}$ is a $\theta$-edge.}
\label{fig:gabriel-2}
\end{center}
\end{figure}

%%%%%%%%%%%%%%%%%%%%%%%%%%%%%%%%%%%%%
%%%%%%%%%%%%%%%%%%%%%%%%%%%%%%%%%%%%%
%%%%%%%%%%%%%%%%%%%%%%%%%%%%%%%%%%%%%
%%%%%%%%%%%%%%%%%%%%%%%%%%%%%%%%%%%%%
%%%%%%%%%%%%%%%%%%%%%%%%%%%%%%%%%%%%%
%%%%%%%%%%%%%%%%%%%%%%%%%%%%%%%%%%%%%

We now prove Claim~2.  First, we prove that, for {\em every} $\theta$ in the interval $[-45^\circ;45^\circ]$, there exists a maximal $\theta$-path starting at $s$ that is low or high. Indeed, it suffices to prove that there exists a $\theta$-edge incident to $s$, as such an edge is also a $\theta$-path starting at $s$, and the existence of a $\theta$-path starting at $s$ implies the existence of a maximal $\theta$-path starting at $s$. Consider a straight-line segment $e_{\theta}$ that is the intersection of a directed half-line incident to $s$ with slope $\theta$ and of a disk of arbitrarily small radius centered at $s$. If $e_{\theta}$ is internal to ${\cal H}_P$, then consider the two edges $e_1$ and $e_2$ of $G$ that are encountered when counter-clockwise and clockwise rotating $e_{\theta}$ around $s$, respectively. Then, $e_1$ or $e_2$ is a $\theta$-edge, as the angle spanned by a clockwise rotation bringing $e_1$ to coincide with $e_2$ is smaller than $90^\circ$, given that $G$ is a Gabriel triangulation, and $e_{\theta}$ is encountered during such a rotation. If $e_{\theta}$ is outside ${\cal H}_P$, which might happen if $s$ on the boundary of ${\cal H}_P$, then assume that the slope of $e_{\theta}$ is in the interval $[0^\circ;45^\circ]$ (the case in which the slope of $e_{\theta}$ is in the interval $[-45^\circ;0^\circ]$ is analogous). Then, the angle spanned by a clockwise rotation bringing $e_{\theta}$ to coincide with $\overline{st}$ is  at most $45^\circ$. Since $\overline{st}$ is in interior or on the boundary of ${\cal H}_P$, an edge $e_1$ of $G$ is encountered during such a rotation, hence $e_1$ is a $\theta$-edge. An analogous proof shows that, for {\em every} $\theta$ in the interval $[-45^\circ;45^\circ]$, there exists a maximal $(\theta+180^\circ)$-path starting at $t$ that is low or high.

Second, we prove that, for {\em some} $\theta\in[-45^\circ;45^\circ]$, there exist a maximal low $\theta$-path {\em and} a maximal high $\theta$-path both starting at $s$. All the maximal $(-45^\circ)$-paths (all the maximal $(45^\circ)$-paths) starting at $s$ are low (resp. high), given that every edge on these paths has slope in the interval $[-90^\circ;0^\circ]$ (resp. $[0^\circ;90^\circ]$). Thus, let $\theta$ be the smallest constant in the interval $[-45^\circ;45^\circ]$ such that a maximal high $\theta$-path exists. We prove that there also exists a maximal low $\theta$-path starting at $s$. Consider an arbitrarily small $\epsilon>0$. By assumption, there exists no high $(\theta-\epsilon)$-path. Hence, from the previous argument there exists a low $(\theta-\epsilon)$-path $\cal P$. If $\epsilon$ is sufficiently small, then no edge of $\cal P$ has slope in the interval $[\theta-45^\circ-\epsilon;\theta-45^\circ)$. Thus every edge of $\cal P$ has slope in the interval $[\theta-45^\circ;\theta+45^\circ-\epsilon)$, hence $\cal P$ is a maximal low $\theta$-path starting at $s$.

Since there exist a maximal high $\theta$-path starting at $s$, a maximal low $\theta$-path starting at $s$, and a maximal $(\theta+180^\circ)$-path starting at $t$ that is low or high, it follows that there exist a maximal $\theta$-path ${\cal P}_s$ starting at $s$ and a maximal $(\theta+180^\circ)$-path ${\cal P}_t$ starting at $t$ that are both high or both low. This proves Claim~2 and hence the lemma.
\end{proof}

%%%%%%%%%%%%%%%%%%%%%%%%%%%%%%%%%%%%%%%%%%%%%
%%%%%%%%%%%%%%%%%%%%%%%%%%%%%%%%%%%%%%%%%%%%%
%%%%%%%%%%%%%%%%%%%%%%%%%%%%%%%%%%%%%%%%%%%%%
%%%%%%%%%%%%%%%%%%%%%%%%%%%%%%%%%%%%%%%%%%%%%
%%%%%%%%%%%%%%%%%%%%%%%%%%%%%%%%%%%%%%%%%%%%%
%%%%%%%%%%%%%%%%%%%%%%%%%%%%%%%%%%%%%%%%%%%%%
%%%%%%%%%%%%%%%%%%%%%%%%%%%%%%%%%%%%%%%%%%%%%

Lemma~\ref{le:theta-path-approaching} and Lemma~\ref{le:gabriel-is-theta} immediately imply the following.

\begin{corollary}\label{cor:gabriel-is-increasing}
Any Gabriel triangulation is increasing-chord.
\end{corollary}

We are now ready to state the main result of this section.

\begin{theorem} \label{th:steiner}
Let $P$ be a point set with $n$ points.
%There exists
One can construct in $O(n \log n)$ time an increasing-chord planar graph $G(P',S)$ such that $P\subseteq P'$ and $|P'|\in O(n)$.
\end{theorem}
\begin{proof}
Bern, Eppstein, and Gilbert~\cite{beg-pgmg-94} proved that, for any point set $P$, there exists a point set $P'$ with $P\subseteq P'$ and $|P'|\in O(n)$ such that $P'$ admits a Gabriel triangulation $G$. Both $P'$ and $G$ can be computed in $O(n \log n)$ time~\cite{beg-pgmg-94}. By Corollary~\ref{cor:gabriel-is-increasing}, $G$ is increasing-chord, which concludes the proof.
\end{proof}

We remark that $o(|P|)$ Steiner points are not always enough to augment a point set $P$ to a point set that admits a Gabriel triangulation. Namely, consider any point set $B$ with $O(1)$ points that admits no Gabriel triangulation. Construct a point set $P$ out of $|P|/|B|$ copies of $B$ placed ``far apart'' from each other, so that any triangle with two points in different copies of $B$ is obtuse. Then, a Steiner point has to be added inside the convex hull of each copy of $B$ to obtain a point set that admits a Gabriel triangulation.

%%%%%%%%%%%%%%%%%%%%%%%%%%%%%%%%%%%
%%%%%%%%%%%%%%%%%%%%%%%%%%%%%%%%%%%
%%%%%%%%%%%%%%%%%%%%%%%%%%%%%%%%%%%
%%%%%%%%%%%%%%%%%%%%%%%%%%%%%%%%%%%
%%%%%%%%%%%%%%%%%%%%%%%%%%%%%%%%%%%
%%%%%%%%%%%%%%%%%%%%%%%%%%%%%%%%%%%
%%%%%%%%%%%%%%%%%%%%%%%%%%%%%%%%%%%
%%%%%%%%%%%%%%%%%%%%%%%%%%%%%%%%%%%
%%%%%%%%%%%%%%%%%%%%%%%%%%%%%%%%%%%
%%%%%%%%%%%%%%%%%%%%%%%%%%%%%%%%%%%
%%%%%%%%%%%%%%%%%%%%%%%%%%%%%%%%%%%
%%%%%%%%%%%%%%%%%%%%%%%%%%%%%%%%%%%
%%%%%%%%%%%%%%%%%%%%%%%%%%%%%%%%%%%
%%%%%%%%%%%%%%%%%%%%%%%%%%%%%%%%%%%
%%%%%%%%%%%%%%%%%%%%%%%%%%%%%%%%%%%

\section{Increasing-Chord Convex Graphs with Few Edges}

In this section we prove the following theorem;

\begin{theorem} \label{th:convex-non-planar}
For every convex point set $P$ with $n$ points, there exists an increasing-chord geometric graph $G(P,S)$ such that $|S|\in O(n \log n)$.
\end{theorem}

The main idea behind the proof of Theorem~\ref{th:convex-non-planar} is that any convex point set $P$ can be decomposed into some one-sided convex point sets $P_1,\dots,P_k$ (which by Lemma~\ref{lemma:half_convex_planar_graph} admit increasing-chord spanning graphs with linearly many edges) in such a way that every two points of $P$ are part of some $P_i$ and that $\sum |P_i|$ is small. In order to perform such a decomposition, we introduce the concept of {\em balanced $({\vec d}_1,{\vec d}_2)$-partition}.

Let $P$ be a convex point set and let $\vec d$ be a directed straight line not orthogonal to any line through two points of $P$. See Fig.~\ref{fig:sets}. Let $p_a({\vec d})$ and $p_b({\vec d})$ be the minimum and maximum point of $P$ with respect to ${\vec d}$, respectively. Let $P_1({\vec d})$ be composed of those points in $P$ that are encountered when clockwise walking along the boundary of ${\cal H}_P$ from $p_a({\vec d})$ to $p_b({\vec d})$, where $p_a({\vec d})\in P_1({\vec d})$ and $p_b({\vec d})\notin P_1({\vec d})$. Analogously, let $P_2({\vec d})$ be composed of those points in $P$ that are encountered when clockwise walking along the boundary of ${\cal H}_P$ from $p_b({\vec d})$ to $p_a({\vec d})$, where $p_b({\vec d})\in P_2({\vec d})$ and $p_a({\vec d})\notin P_2({\vec d})$.

\begin{figure}[tb]
\begin{center}
\mbox{\includegraphics[scale=0.3]{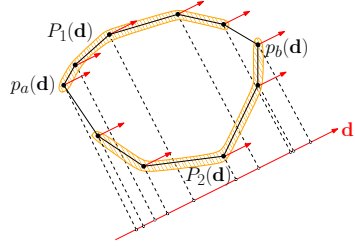}}
\caption{Subsets $P_1({\vec d})$ and $P_2({\vec d})$ of a point set $P$ determined by a directed straight line $\vec d$.}
\label{fig:sets}
\end{center}
\end{figure}

Let ${\vec d}_1$ and ${\vec d}_2$ be two directed straight lines not orthogonal to any line through two points of $P$, where the clockwise rotation that brings ${\vec d}_1$ to coincide with ${\vec d}_2$ is at most $180^\circ$. The {\em $({\vec d}_1,{\vec d}_2)$-partition} of $P$ partitions $P$ into subsets $P_a=P_{1}({\vec d}_1)\cap P_{1}({\vec d}_2)$, $P_b=P_{1}({\vec d}_1)\cap P_{2}({\vec d}_2)$, $P_c=P_{2}({\vec d}_1)\cap P_{1}({\vec d}_2)$, and
$P_d=P_{2}({\vec d}_1)\cap P_{2}({\vec d}_2)$. Note that every point in $P$ is contained in one of $P_a$, $P_b$, $P_c$, and $P_d$. A {\em $({\vec d}_1,{\vec d}_2)$-partition} of $P$ is {\em balanced} if $|P_a|+|P_d|\leq \frac{|P|}{2} +1$ and $|P_b|+|P_c|\leq \frac{|P|}{2} +1$. We now argue that, for every point set $P$, a balanced $({\vec d}_1,{\vec d}_2)$-partition of $P$ always exists, even if ${\vec d}_1$ is arbitrarily prescribed.

\begin{lemma} \label{le:partition}
Let $P$ be a convex point set and let ${\vec d}_1$ be a directed straight line not orthogonal to any line through two points of $P$. Then, there exists a directed straight line ${\vec d}_2$ that is not orthogonal to any line through two points of $P$ such that the $({\vec d}_1,{\vec d}_2)$-partition of $P$ is balanced.
\end{lemma}

\begin{proof}
Denote by $q_1=p_a({\vec d}_1),q_2,\dots,q_l,q_{l+1}=p_b({\vec d}_1)$ the points of $P$ encountered when clockwise walking on the boundary of ${\cal H}_P$ from $p_a({\vec d}_1)$ to $p_b({\vec d}_1)$. Also, denote by $r_1=p_b({\vec d}_1),r_2,\dots,r_m,r_{m+1}=p_a({\vec d}_1)$ the points of $P$ encountered when clockwise walking on the boundary of ${\cal H}_P$ from $p_b({\vec d}_1)$ to $p_a({\vec d}_1)$.

Initialize ${\vec d}_2$ to be a directed straight line coincident with ${\vec d}_1$. When ${\vec d}_2={\vec d}_1$, we have $P_a=\{q_1,q_2,\dots,q_l\}$, $P_d=\{r_1,r_2,\dots,r_m\}$, $P_b=\emptyset$, and $P_c=\emptyset$. We now clockwise rotate ${\vec d}_2$ until it is opposite to ${\vec d}_1$ (that is, parallel and pointing in the opposite direction). As we rotate ${\vec d}_2$, sets $P_{1}({\vec d}_2)$ and $P_{2}({\vec d}_2)$ change, hence sets $P_a$, $P_b$, $P_c$, and $P_d$ change as well. When ${\vec d}_2$ is opposite to ${\vec d}_1$, we have $P_a=\emptyset$, $P_d=\emptyset$, $P_b=\{q_1,q_2,\dots,q_l\}$, and $P_c=\{r_1,r_2,\dots,r_m\}$. We will argue that there is a moment during such a rotation of ${\vec d}_2$ in which the corresponding $({\vec d}_1,{\vec d}_2)$-partition of $P$ is balanced. Assume that at any time instant during the rotation of ${\vec d}_2$ the following hold (see Figs.~\ref{fig:rotating}(a)--(b)):

\begin{itemize}
\item $P_b=\{q_1,q_2,\dots,q_j\}$ (possibly $P_b$ is empty);
\item $P_a=\{q_{j+1},q_{j+2},\dots,q_l\}$ (possibly $P_a$ is empty);
\item $P_c=\{r_1,r_2,\dots,r_k\}$ (possibly $P_c$ is empty);
\item $P_d=\{r_{k+1},r_{k+2},\dots,r_m\}$ (possibly $P_d$ is empty); and
\item $q_{j+1}$ and $r_{k+1}$ are the minimum and maximum point of $P$ w.r.t. ${\vec d}_2$, respectively.
\end{itemize}

\begin{figure}[tb]
\begin{center}
\begin{tabular}{c c}
\mbox{\includegraphics[scale=0.3]{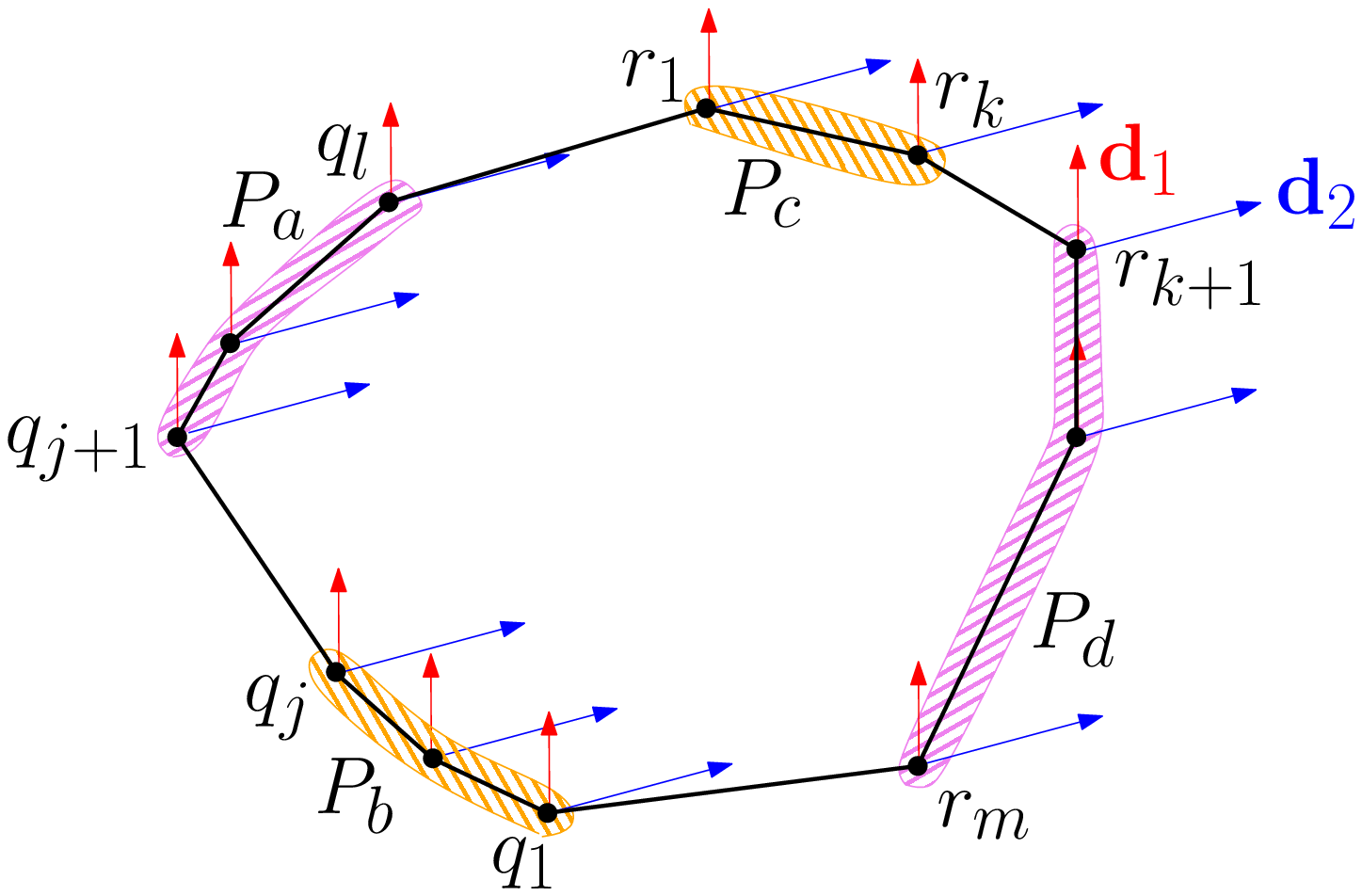}} & \hspace{5mm}
\mbox{\includegraphics[scale=0.35]{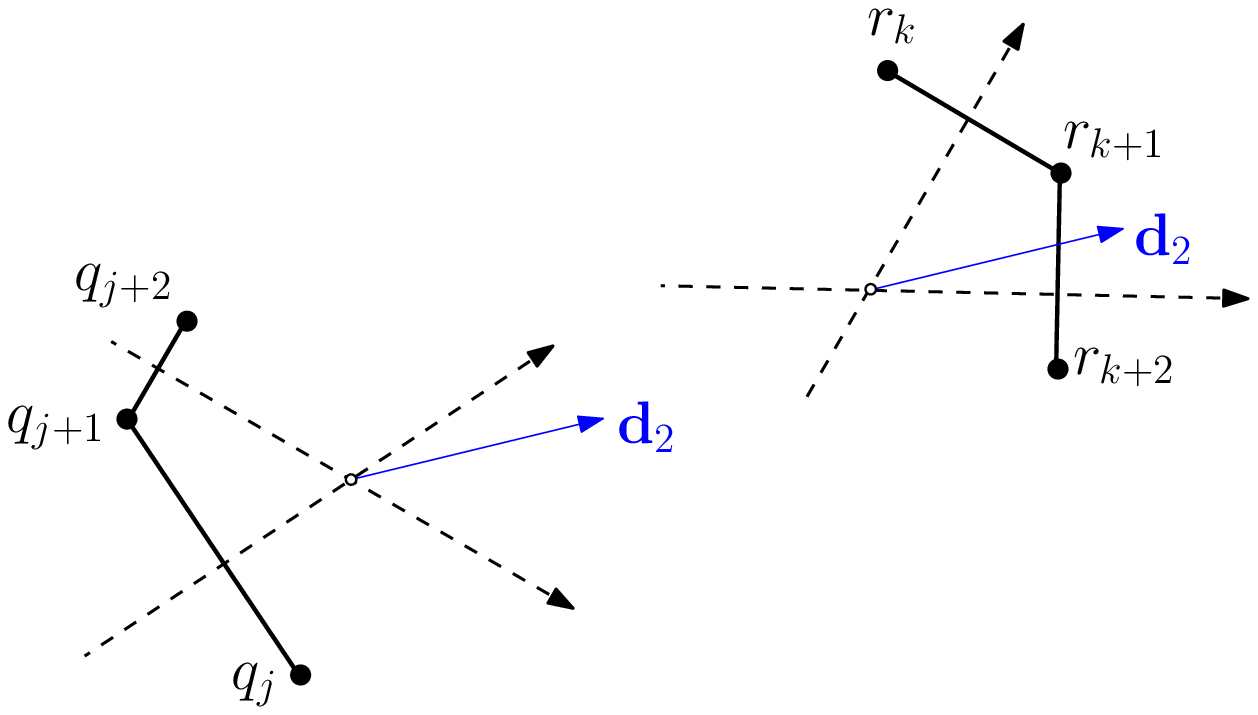}}\\
(a) \hspace{5mm} & \hspace{5mm} (b)
\end{tabular}
\caption{(a) Sets $P_a$, $P_b$, $P_c$, and $P_d$ at a certain time instant during the rotation of ${\vec d}_2$. (b) The slope of ${\vec d}_2$ with respect to the slopes of the lines orthogonal to $\overline{q_{j}q_{j+1}}$, to $\overline{q_{j+1}q_{j+2}}$, to $\overline{r_{k}r_{k+1}}$, and to $\overline{r_{k+1}r_{k+2}}$.}
\label{fig:rotating}
\end{center}
\end{figure}

The assumption is indeed true when ${\vec d}_2$ starts moving, with $j=0$ and $k=0$.

As we keep on clockwise rotating ${\vec d}_2$, at a certain moment ${\vec d}_2$ becomes orthogonal to $\overline{q_{j+1}q_{j+2}}$ or to $\overline{r_{k+1}r_{k+2}}$ (or to both if $\overline{q_{j+1}q_{j+2}}$ and $\overline{r_{k+1}r_{k+2}}$ are parallel). Thus, as we keep on clockwise rotating ${\vec d}_2$, sets $P_a$, $P_b$, $P_c$, and $P_d$ change. Namely:

If ${\vec d}_2$ becomes orthogonal first to $\overline{q_{j+1}q_{j+2}}$ and then to $\overline{r_{k+1}r_{k+2}}$, then as ${\vec d}_2$ rotates clockwise after the position in which it is orthogonal to $\overline{q_{j+1}q_{j+2}}$, we have

\begin{itemize}
\item $P_b=\{q_1,q_2,\dots,q_j,q_{j+1}\}$;
\item $P_a=\{q_{j+2},q_{j+3},\dots,q_l\}$ (possibly $P_a$ is empty);
\item $P_c=\{r_1,r_2,\dots,r_k\}$ (possibly $P_c$ is empty);
\item $P_d=\{r_{k+1},r_{k+2},\dots,r_m\}$ (possibly $P_d$ is empty); and
\item $q_{j+2}$ and $r_{k+1}$ are the minimum and maximum point of $P$ w.r.t. ${\vec d}_2$, respectively.
\end{itemize}

%%%%%%%%%%%%%%%%%%%%%%%%%%%%%%%

If ${\vec d}_2$ becomes orthogonal first to $\overline{r_{k+1}r_{k+2}}$ and then to $\overline{q_{j+1}q_{j+2}}$, then as ${\vec d}_2$ rotates clockwise after the position in which it is orthogonal to $\overline{r_{k+1}r_{k+2}}$, we have that $P_a$ and $P_b$ stay unchanged, that $r_{k+1}$ passes from $P_d$ to $P_c$, and that $q_{j+1}$ and $r_{k+2}$ are the minimum and maximum point of $P$ w.r.t. ${\vec d}_2$, respectively.

%%%%%%%%%%%%%%%%%%%%%%%%%%%%%%%

If ${\vec d}_2$ becomes orthogonal to $\overline{q_{j+1}q_{j+2}}$ and $\overline{r_{k+1}r_{k+2}}$ simultaneously, then as ${\vec d}_2$ rotates clockwise after the position in which it is orthogonal to  $\overline{q_{j+1}q_{j+2}}$, we have that $q_{j+1}$ passes from $P_a$ to $P_b$, that $r_{k+1}$ passes from $P_d$ to $P_c$, and that $q_{j+2}$ and $r_{k+2}$ are the minimum and maximum point of $P$ w.r.t. ${\vec d}_2$, respectively.

Observe that:
\begin{enumerate}
\item whenever sets $P_a$, $P_b$, $P_c$, and $P_d$ change, we have that $|P_a|+|P_d|$ and $|P_b|+|P_c|$ change at most by two;
\item when ${\vec d}_2$ starts rotating we have that $|P_a|+|P_d|=|P|$, and when ${\vec d}_2$ stops rotating we have that $|P_a|+|P_d|=0$;
\item when ${\vec d}_2$ starts rotating we have that  $|P_b|+|P_c|=0$, and when ${\vec d}_2$ stops rotating we have that $|P_b|+|P_c|=|P|$; and
\item $|P_a|+|P_b|+|P_c|+|P_d|=|P|$ holds at any time instant.
\end{enumerate}

By continuity, there is a time instant in which $|P_a|+|P_d|=\lfloor |P|/2\rfloor$ and $|P_b|+|P_c|=\lceil |P|/2 \rceil$, or in which $|P_a|+|P_d|=\lfloor |P|/2\rfloor +1$ and $|P_b|+|P_c|=\lceil |P|/2\rceil-1$. This completes the proof of the lemma.
\end{proof}

We now show how to use Lemma~\ref{le:partition} in order to prove Theorem~\ref{th:convex-non-planar}.

Let $P$ be any point set. Assume that no two points of $P$ have the same $y$-coordinate. Such a condition is easily met after rotating the Cartesian axes. Denote by $\vec l$ a vertical straight line directed towards increasing $y$-coordinates. Each of $P_1({\vec l})$ and $P_2({\vec l})$ is convex and one-sided with respect to $\vec l$. By Lemma~\ref{lemma:half_convex_planar_graph}, there exist increasing-chord graphs $G_1=(P_1({\vec l}),S_1)$ and $G_2=(P_2({\vec l}),S_2)$ with $|S_1|<2|P_1({\vec l})|$ and $|S_2|<2|P_2({\vec l})|$. Then, graph $G(P,S_1\cup S_2)$ has less than $2(|P_1({\vec l})|+|P_2({\vec l})|)=2|P|$ edges and contains an increasing-chord path between every pair of vertices in $P_1({\vec l})$ and between every pair of vertices in $P_2({\vec l})$. However, $G$ does not have increasing-chord paths between any pair $(a,b)$ of vertices such that $a\in P_1({\vec l})$ and $b\in P_2({\vec l})$.

We now present and prove the following claim. Consider a convex point set $Q$ and a directed straight line ${\vec d}_1$ not orthogonal to any line through two points of $Q$. Then, there exists a geometric graph $H(Q,R)$ that contains an increasing-chord path between every point in $Q_1({\vec d}_1)$ and every point in $Q_2({\vec d}_1)$, such that $|R|\in O(|Q|\log |Q|)$.

The application of the claim with $Q=P$ and ${\vec d}_1={\vec l}$ provides a graph $H(P,R)$ that contains an increasing-chord path between every pair $(a,b)$ of vertices such that $a\in P_1({\vec l})$ and $b\in P_2({\vec l})$. Thus, the union of $G$ and $H$ is an increasing-chord graph with $O(|P|\log |P|)$ edges spanning $P$. Therefore, the above claim implies Theorem~\ref{th:convex-non-planar}.

We show an inductive algorithm to construct $H$. Let $f(Q,{\vec d}_1)$ be the number of edges that $H$ has as a result of the application of our algorithm on a point set $Q$ and a directed straight-line ${\vec d}_1$. Also, let $f(n)=\max\{f(Q,{\vec d}_1)\}$, where the maximum is among all point sets $Q$ with $n=|Q|$ points and among all the directed straight-lines ${\vec d}_1$ that are not orthogonal to any line through two points of $Q$.

Let $Q$ be any convex point set with $n$ points and let ${\vec d}_1$ be any directed straight line not orthogonal to any line through two points of $Q$. By Lemma~\ref{le:partition}, there exists a directed straight line not orthogonal to any line through two points of $Q$ and such that the $({\vec d}_1,{\vec d}_2)$-partition of $Q$ is balanced.

Let $Q_{a}=Q_{1}({\vec d}_1)\cap Q_{1}({\vec d}_2)$, let $Q_{b}=Q_{1}({\vec d}_1)\cap Q_{2}({\vec d}_2)$, let $Q_{c}=Q_{2}({\vec d}_1)\cap Q_{1}({\vec d}_2)$, and let $Q_{d}=Q_{2}({\vec d}_1)\cap Q_{2}({\vec d}_2)$.

Point set $Q_a\cup Q_c$ is convex and one-sided with respect to ${\vec d}_2$. By Lemma~\ref{lemma:half_convex_planar_graph} there exists an increasing-chord graph $H_1(Q_a\cup Q_c,R_1)$ with $|R_1|< 2(|Q_a|+|Q_c|)$ edges. Analogously, by Lemma~\ref{lemma:half_convex_planar_graph} there exists an increasing-chord graph $H_2(Q_b\cup Q_d,R_2)$ with $|R_2|<2(|Q_b|+|Q_d|)$ edges.

Hence, there exists a graph $H_3(Q,R_1\cup R_2)$ with $|R_1\cup R_2|< 2(|Q_a|+|Q_c|+|Q_b|+|Q_d|)=2|Q|=2n$ edges containing an increasing-chord path between every point in $Q_a$ and every point in $Q_c$, and between every point in $Q_b$ and every point in $Q_d$. However, $G$ does not have an increasing-chord path between any point in $Q_a$ and any point in $Q_d$, and does not have an increasing-chord path between any point in $Q_b$ and any point in $Q_c$.

By Lemma~\ref{le:partition}, it holds that $|Q_a|+|Q_d|\leq \frac{n}{2}+1$ and $|Q_b|+|Q_d|\leq \frac{n}{2}+1$. By definition, we have $f(Q_a\cup Q_d,{\vec d}_1)\leq f(|Q_a|+|Q_d|)\leq f(\frac{n}{2}+1)$. Analogously, it holds that $f(Q_b\cup Q_c,{\vec d}_1)\leq f(|Q_b|+|Q_c|)\leq f(\frac{n}{2}+1)$. Hence, $f(n)\leq 2n+2f(\frac{n}{2}+1)\in O(n\log n)$. This proves the claim and hence Theorem~\ref{th:convex-non-planar}.

\section{Conclusions}

We considered the problem of constructing increasing-chord graphs spanning point sets. We proved that, for every point set $P$, there exists a planar increasing-chord graph $G(P',S)$ with $P\subseteq P'$ and $|P'|\in O(|P|)$. We also proved that, for every convex point set $P$, there exists an increasing-chord graph $G(P,S)$ with $|S|\in O(|P| \log |P|)$.

Despite our research efforts, the main question on this topic remains open:

\begin{problem}
Is it true that, for every (convex) point set $P$, there exists an increasing-chord planar graph $G(P,S)$?
\end{problem}

One of the directions we took in order to tackle this problem is to assume that the points in $P$ lie on a constant number of straight lines. While a simple modification of the proof of Lemma~\ref{lemma:half_convex_planar_graph} allows us to prove that an increasing-chord planar graph always exists spanning a set of points lying on two straight lines, it is surprising and disheartening that we could not prove a similar result for sets of points lying on three straight lines. The main difficulty seems to lie in the construction of planar increasing-chord graphs spanning sets of points lying on the boundary of an acute triangle.

Gabriel graphs naturally generalize to higher dimensions, where empty balls replace empty disks. In Section~\ref{se:steiner} we showed that, for points in $\mathbb R^2$, every Gabriel triangulation is increasing-chord. Can this result be generalized to higher dimensions?

\begin{problem}
Is it true that, for every point set $P$ in $\mathbb R^d$, any Gabriel triangulation of $P$ is increasing-chord?\end{problem}

Finally, it would be interesting to understand if increasing-chord graphs with few edges can be constructed for any (possibly non-convex) point set:

\begin{problem}
Is it true that, for every point set $P$, there exists an increasing-chord graph $G(P,S)$ with $|S|\in o(|P|^2)$?
\end{problem}

\bibliographystyle{splncs_srt}
\bibliography{main}

\end{document}